\documentclass[12pt,reqno]{amsart}
\usepackage{amsmath, amsthm,  amsfonts, amscd, epsfig, url}
\usepackage{eurosym}
\usepackage{amsfonts, amsmath,amsthm,amssymb}
\usepackage{amscd,verbatim}
\usepackage{amsmath}
\usepackage{amssymb}
\usepackage[arrow,matrix,curve]{xy}
\usepackage{fullpage}

\usepackage{hyperref}
\hypersetup{colorlinks=true}
\newtheorem{theorem}{Theorem}[section]
\newtheorem{theorem*}{Theorem}

\newtheorem{cor*}{Corollary}

\theoremstyle{definition}

\newtheorem{definition*}{Definition}

\theoremstyle{remark}

\newtheorem{remarks*}{Remarks}

\numberwithin{equation}{section}

\newcommand{\sfU}{{\mathsf U}}

\newcommand{\sfG}{{\mathsf G}}

\newcommand\cE{\mathcal{E}}
\newcommand\cF{\mathcal{F}}
\newcommand\cH{\mathcal{H}}
\newcommand\cK{\mathcal{K}}

\newcommand\cP{\mathcal{P}}

\newcommand\CC{\mathbb C}

\newcommand\RR{\mathbb R}

\newcommand\ZZ{\mathbb Z}

\newcommand{\sfSU}{{\mathsf{SU}}}

\newcommand{\cR}{{\mathcal R}}
\begin{document}

\title[Spherical T-Duality and the spherical Fourier-Mukai transform]
{Spherical T-Duality and\\ the spherical Fourier-Mukai transform}

\author[P Bouwknegt]{Peter Bouwknegt}

\address[Peter Bouwknegt]{
Mathematical Sciences Institute, and
Department of Theoretical Physics,
Research School of Physics and Engineering,
The Australian National University,
Canberra, ACT 2601, Australia}

\email{peter.bouwknegt@anu.edu.au}

\author[J Evslin]{Jarah Evslin}

\address[Jarah Evslin]{
High Energy Nuclear Physics Group,
Institute of Modern Physics,
Chinese Academy of Sciences,
Lanzhou, China}

\email{jarah@impcas.ac.cn}

\author[V Mathai]{Varghese Mathai}

\address[Varghese Mathai]{
Department of Pure Mathematics,
School of  Mathematical Scienes,
University of Adelaide,
Adelaide, SA 5005,
Australia}

\email{mathai.varghese@adelaide.edu.au}

% ===================================================================

\begin{abstract}
In \cite{BEM14,BEM142}, we introduced spherical T-duality, which relates pairs of the form $(P,H)$ consisting of an oriented
$S^3$-bundle $P\rightarrow M$ and a 7-cocycle $H$ on $P$ called the 7-flux.  Intuitively, the spherical T-dual is another such pair $(\hat P, \hat H)$ and spherical T-duality exchanges the 7-flux
with the Euler class, upon fixing the Pontryagin class and the second Stiefel-Whitney class.
Unless $\mathrm{dim}(M)\leq 4$, not all pairs admit spherical T-duals and the spherical T-duals are not always unique. 
In this paper, we define a canonical
Poincar\'e virtual line bundle $\cP$ on $S^3 \times S^3$ (actually also for $S^n\times S^n$)
and the spherical Fourier-Mukai transform, which
 implements a degree shifting
isomorphism in K-theory on the trivial $S^3$-bundle. This is then used
 to prove that all spherical T-dualities induce natural degree-shifting isomorphisms between  the 7-twisted K-theories
of the pairs $(P,H)$ and $(\hat P, \hat H)$ when $\mathrm{dim}(M)\leq 4$, improving our earlier results.  
\end{abstract}

%\thanks{
%JE is supported by NSFC MianShang grant 11375201. PB and VM thank the
%Australian Research Council for support via ARC Discovery Project grants
%DP150100008 and DP130103924.}



\maketitle

% ===================================================================
\section{Introduction}

Recall that the renowned Poincar\'e line bundle $\cP\to S^1\times S^1$ is tautologically defined
and comes with a canonical connection whose curvature is the standard symplectic 2-form on
$S^1\times S^1$. More generally, it is defined in the holomorphic context
on a polarised abelian variety in Mumford \cite{Mum70},  chapters 10-13,
where it was used to study fine moduli problems. It was then used by Mukai \cite{Mukai} to give
an equivalence of derived categories of coherent sheaves on an abelian variety with its dual abelian variety.
In the smooth context, Hori \cite{Hori99} used the Poincar\'e line bundle to give a (shifted) equivalence
of K-theories, and thereby establishing the equivalence of charges in type IIA and type IIB string theories
in the absence of background fluxes. In \cite{BEM,BEM2} (see also \cite{BS}) a deep extension was made
for principal torus bundles with nontrivial fluxes, where an equivalence of twisted K-theories was derived
but importantly that there was a change in spacetime topology in general for the first time.

In this paper (section \ref{sect:poincare}) we define a Poincar\'e virtual line bundle
$$\cP \longrightarrow \sfSU(2)\times \sfSU(2)$$
for the first time, 
%making our discussion of spherical T-duality almost on par with the torus case.
and it represents the diagonal class in K-theory and
implements a canonical equivalence of K-theories in the case of trivial $\sfSU(2)$-bundles as shown in section \ref{sect:FM}.
This can be viewed an an analog of Hori's result that was mentioned above. All this is generalised from $\sfSU(2)=S^3$ to general spheres $S^n$ 
in section  \ref{sect:Tduality}.

In \cite{BEM14,BEM142}, we introduced a new kind of duality for string theory (M-theory), termed spherical T-duality,
for 7D spacetimes that are compactified as $\sfSU(2)$-bundles with 7-flux over 4D manifolds,
\begin{equation}\label{SU(2)4}
\begin{CD}
\sfSU(2) @>>> \,  P \\
&& @V \pi VV \\
&& M \end{CD}
\end{equation}
In \cite{BEM14} we dealt only with principal $\sfSU(2)$-bundles with 7-flux, whereas in \cite{BEM142} we dealt with 
the more general case of (oriented) $\sfSU(2)$-bundles  with 7-flux that were not necessarily principal bundles. 
We showed that a principal $\sfSU(2)$-bundle with 7-flux, had a unique spherical T-dual principal $\sfSU(2)$-bundle with 
T-dual 7-flux, where the 7-flux gets exchanged with the 2nd Chern number, and there is an equivalence of 7-twisted
cohomologies and 7-twisted K-theories (modulo an extension problem), see  \cite{BEM14}. On the other hand, a 
non-principal (oriented) $\sfSU(2)$-bundle  with 7-flux can have infinitely many spherical T-duals that are also
non-principal (oriented) $\sfSU(2)$-bundles  with 7-flux, and once again there is an equivalence of 7-twisted
cohomologies and 7-twisted K-theories (modulo an extension problem), see  \cite{BEM142}. The problem in this case is that 
non-principal (oriented) $\sfSU(2)$-bundles on simply-connected 4-manifolds are classified by the 2nd Chern number (or Euler number) as well as
the Pontryagin number and the 2nd Stiefel-Whitney class. So if we fix the the Pontryagin number, and the 2nd Stiefel-Whitney class,
we again get a unique spherical T-dual in the non-principal case also.  
Since the initial version of this paper was posted on the arxiv, the interesting paper \cite{LSW} appeared on the arxiv, which formulated a generalization of spherical T-duality 
to possibly nonorientable sphere bundles, and proved that it induces an isomorphism on a class of twisted cohomology theories that
includes algebraic K-theory.

In  \cite{BEM14,BEM142}, we argued that the
7-twisted cohomology and the 7-twisted K-theory which featured in our main theorems classify certain conserved charges
in type IIB supergravity. We concluded that spherical T-duality provides a one to one map between conserved 
charges in certain topologically distinct compactifications and also a novel electromagnetic duality on the fluxes.
%In this paper, we give a general proof of this fact, utilising the construction of the spherical Poincar\'e (virtual) line bundle.

In section \ref{sect:Tduality}, we show that the Poincar\'e  virtual line bundle gives rise
to isomorphisms of 7-twisted K-theories for 7-dimensional principal $\sfSU(2)$-bundles with 7-fluxes, 
significantly improving our earlier results which proved this only modulo an extension problem.  In this version of our paper, we use the higher sphere formalism of Ref.~\cite{LSW} to generalize this isomorphism to $S^n$ bundles with $(2n+1)$-flux.
We also compute the spherical T-duality group, the twisted K-theories of $S^3$ bundles over simply connected, oriented four manifolds and speculate on links with String theory/M-theory.
%oflinks with String theory/M-theory.
%Moreover we compute the Chern-Weil
%representative of the 3rd Chern class of $\cP$. All of this strengthens certain results in  \cite{BEM14}.

\tableofcontents

% ===================================================================
%\section{Spherical T-duality for trivial $\sfSU(2)$-bundles and the spherical Fourier-Mukai transform}

\section{Poincar\'e element and spherical Fourier-Mukai transform in K-theory}\label{sect:FM}
%\medskip
In this section, we define Poincar\'e element and spherical Fourier-Mukai transform on the K-theory of {\em trivial} $\sfSU(2)$-bundles.
This can be generalised to $S^n$ bundles in a striaghtforward way. The spherical Fourier-Mukai transform exchanges the tensor
product and convolution operations on these K-theories, showing that it is the geometric analog of the usual Fourier transform.
In later sections, we will define the spherical Fourier-Mukai transform for {\em non-trivial} $\sfSU(2)$-bundles over 4D manifolds 
with 7-flux.

The {\em Poincar\'e element} $[\cP]$ over $\sfSU(2) \times \widehat{\sfSU(2)}$, where
$\widehat{\sfSU(2)}= \sfSU(2)$,  is
the diagonal class in
$$
K^0(\sfSU(2) \times \widehat{\sfSU(2)}) \cong K^0(\sfSU(2)) \otimes K^0(\widehat{\sfSU(2)}) \oplus  
K^1(\sfSU(2)) \otimes K^1(\widehat{\sfSU(2)})\,,
$$
that is $[\cP] = 1 \otimes \widehat 1 + \zeta \otimes \widehat \zeta$,
where $\zeta \in K^1(\sfSU(2))$ and $\widehat \zeta \in K^1(\widehat{\sfSU(2)})$ are the generators,
represented by degree 1 maps $\sfSU(2)\longmapsto \mathsf{U}(N), \, N\gg 0$. Later on, we will describe
a canonical   vector bundle representative of $[\cP]$.\\

Consider the trivial $\sfSU(2)$-bundle $P=M\times \sfSU(2)$. Consider the commutative diagram

\begin{equation*}
\xymatrix{  & M \times \sfSU(2) \times \widehat{\sfSU(2)} \ar[dl]_{\widehat p} \ar[dr]^{p} &  \\
P=M \times \sfSU(2) \ar[dr]_\pi  && M \times \widehat{\sfSU(2)} = \widehat{P}  \ar[dl]^{\widehat \pi} \\
& M &  }
\end{equation*}

\medskip

\begin{theorem} For $E$ a vector bundle over $P$, define the spherical Fourier-Mukai transform as
\begin{equation*}
\cF [E] = p_* \, ( \widehat p^*\,  [E] \otimes [\cP] )\,,
\end{equation*}
giving rise to the spherical Fourier-Mukai transform in (compactly supported) K-theory,
\begin{equation*} \begin{CD}
\cF\ :\ K^i_c(P) @>\cong >> K^{i+1}_c(\widehat{P})\,.
\end{CD}
\end{equation*}
\end{theorem}

\begin{proof}
By the K\"unneth theorem,
$$
K^0_c(P) \cong K^0_c(M) \oplus K^1_c(M) \cong K^1_c(P)\,,
$$
and similarly
$$
K^0_c(\widehat  P) \cong K^0_c(M) \oplus K^1_c(M) \cong K^1_c(\widehat  P)\,.
$$

Now if $x\in K^0_c(P)$, then  $x=x_0 \otimes 1 + x_1 \otimes \zeta$ where $x_j \in K^j_c(M), \, j=0,1$.
Then an easy computation shows that
$$
\cF(x)=\cF(x_0 \otimes 1 + x_1 \otimes \zeta) = x_0 \otimes \widehat\zeta + x_1 \otimes 1 \,.
$$
It follows that
\begin{equation*} \begin{CD}
\cF\ :\ K^0_c(P) @>\cong >> K^{1}_c(\widehat{P})
\end{CD}
\end{equation*}
is an isomorphism.

Similarly, if $x\in K^1_c(P)$, then  $x=x_0 \otimes \zeta + x_1 \otimes 1$ where $x_j \in K^j_c(M), \, j=0,1$.
Then an easy computation shows that
$$
\cF(x)=\cF(x_0 \otimes \zeta + x_1 \otimes 1) = x_0 \otimes 1 + x_1 \otimes \widehat\zeta
$$
It follows that
\begin{equation*} \begin{CD}
\cF\ :\ K^1_c(P) @>\cong >> K^{0}_c(\widehat{P})
\end{CD}
\end{equation*}
is also an isomorphism.

\end{proof}

Define a commutative, associative products on $K^\bullet(\sfSU(2))$ given by
\begin{align*}
1\otimes 1 = 1\,, \qquad 1 \star 1 =0\,,\\
1\otimes \zeta = \zeta\,, \qquad 1 \star \zeta = 1\,,\\
\zeta\otimes \zeta =0\,, \qquad \zeta\star \zeta = \zeta\,,
\end{align*}
called the tensor product and convolution product (induced by the multiplication on $\sfSU(2)$ and Poincar\'e duality), respectively and with the standard compatibility relations making $K^\bullet(\sfSU(2))$ into a bialgebra.
This in turn defines commutative, associative products on $K^\bullet(P)$ and $K^\bullet(\widehat P)$, both equal to 
$K^\bullet(M)\otimes K^\bullet(\sfSU(2))$ and one has the following result, which justifies the nomenclature "spherical Fourier-Mukai 
transform" as it resembles the Fourier transform, and was first defined in the holomorphic context for the Poincar\'e line bundle
by Mukai in \cite{Mukai},

\begin{theorem}
The Fourier-Mukai transform in K-theory
\begin{equation*} \begin{CD}
\cF\ :\ K^i_c(P) @>\cong >> K^{i+1}_c(\widehat{P}) \,,
\end{CD}
\end{equation*}
takes the tensor product to the convolution product and the convolution product to the tensor product.
\end{theorem}

\begin{proof} In the notation above, 
$ (x\otimes 1) \otimes (y\otimes 1) = (x\otimes y) \otimes 1$ therefore $  \cF(  (x\otimes 1) \otimes (y\otimes 1) ) =
\cF ((x\otimes y) \otimes 1) = (x\otimes y) \otimes \zeta$. On the other hand, $ \cF (x\otimes 1) \star  \cF (y\otimes 1)
=  (x\otimes \zeta) \star  (y\otimes \zeta) =  (x\otimes y) \otimes \zeta$.\\

$ (x\otimes 1) \otimes (y\otimes \zeta) = (x\otimes y) \otimes \zeta$ therefore $  \cF(  (x\otimes 1) \otimes (y\otimes \zeta) ) =
\cF ((x\otimes y) \otimes \zeta) = (x\otimes y) \otimes 1$. On the other hand, $ \cF (x\otimes 1) \star  \cF (y\otimes \zeta)
=  (x\otimes \zeta) \star  (y\otimes 1) =  (x\otimes y) \otimes 1$.\\

$ (x\otimes \zeta) \otimes (y\otimes \zeta) = 0$ therefore $  \cF(  (x\otimes 1) \otimes (y\otimes \zeta) ) = 0$.
On the other hand, $ \cF (x\otimes \zeta) \star  \cF (y\otimes \zeta)
=  (x\otimes 1) \star  (y\otimes 1) =  0$.\\

This shows that $ \cF$ takes tensor product to convolution.\\

$ (x\otimes 1) \star (y\otimes 1) = 0$ therefore $  \cF(  (x\otimes 1) \star (y\otimes 1) ) = 0$.
On the other hand, $ \cF (x\otimes 1) \otimes  \cF (y\otimes 1)
=  (x\otimes \zeta) \otimes  (y\otimes \zeta) = 0$.\\

$ (x\otimes 1) \star (y\otimes \zeta) = (x\otimes y) \otimes 1$ therefore $  \cF(  (x\otimes 1) \star (y\otimes \zeta) ) =
\cF ((x\otimes y) \otimes 1) = (x\otimes y) \otimes \zeta$. On the other hand, $ \cF (x\otimes 1) \otimes  \cF (y\otimes \zeta)
=  (x\otimes \zeta) \otimes (y\otimes 1) =  (x\otimes y) \otimes \zeta$.\\

$ (x\otimes \zeta) \star (y\otimes \zeta) = (x\otimes y) \otimes \zeta$ therefore $  \cF(  (x\otimes 1) \star (y\otimes \zeta) ) =
\cF( (x\otimes y) \otimes \zeta) =  (x\otimes y) \otimes 1$.
On the other hand, $ \cF (x\otimes \zeta) \otimes  \cF (y\otimes \zeta)
=  (x\otimes 1) \otimes (y\otimes 1) =   (x\otimes y) \otimes 1$.\\

This shows that $ \cF$ takes convolution to tensor product, completing the proof.

\end{proof}

% ===================================================================
\section{Poincar\'e virtual line bundle}\label{sect:poincare}
Here we give natural vector bundle realisation of the Poincar\'e element of the last section, and call it the Poincar\'e virtual line bundle. 
We also briefly recall higher twisted K-theory and give computations of it for oriented $S^3$-bundles with 7-flux over 
simply connected oriented compact 4D manifolds.

\subsection{Vector bundle realization of the Poincar\'e element}\label{sec:poincare}

From the long exact sequence in homotopy for the principal bundle $\sfSU(2) \to \sfSU(3) \to S^5$, we deduce that
$\pi_5(\sfSU(3))\cong \ZZ$. Let $h: S^5 \to \sfSU(3)$ be a generator, and use it as a clutching function on the equator
of $S^6$ to determine
a principal $\sfSU(3)$-bundle $P$ over $S^6$. In fact, standard arguments in algebraic topology show that 
principal $\sfSU(3)$-bundles $P$ over $S^6$ are classified by the third Chern class $c_3(P)\in 2\ZZ$, cf. \cite{HK}.
Now $[S^3 \times S^3, S^6] \cong H^6(S^3 \times S^3, \ZZ) \cong \ZZ$, so there is a degree 1 map
$g:  S^3 \times S^3 \to S^6$. Then $g^*(P)$ is a principal $\sfSU(3)$-bundle over $S^3 \times S^3$, whose associated complex
vector bundle $\cP$ of rank 3 represents the Poincare object. Note that the restriction of $\cP$  to the submanifolds $S^3\times\{x\}$ and
$\{x\}\times S^3$ are trivializable, similar to the Poincare line bundle on $S^1\times S^1$.

When $P=\mathsf{G}_2$, that is, $\sfSU(3) \to \mathsf{G}_2 \to S^6$, then 
a theorem of Bott says that the top Chern class ${c_n}$ of any (complex) vector bundle on ${S^{2n}}$ is divisible by ${(n-1)!}$, 
so in particular,
$c_3(P)=2$, so that $P$ is one of the bundles that we are searching for.
The associated rank 3 complex vector bundle $\cE$ over $S^6$ is the non-trivial generator of $K^0(S^6)$.
Recall that if $X$ and $Y$ are pointed spaces (i.e. topological spaces with distinguished basepoints $x_0$ and $y_0$)
the wedge sum of $X$ and $Y$, denoted $X\vee Y$, is the quotient space of the disjoint union of $X$ and $Y$ by the identification $x_0 \sim y_0$.
One can think of $X$ and $Y$ as sitting inside $X \times Y$ as the subspaces $X \times \{y_0\}$ and $\{x_0\} \times Y$. 
These subspaces intersect at a single point, $(x_0, y_0)$, the basepoint of $X \times Y$. So the union of these 
subspaces can be identified with the wedge sum $X \vee Y$.
Then the smash product of $X$ and $Y$, denoted $X\wedge Y$ is the quotient space $(X\times Y)/X\vee Y$. In particular, $S^3\wedge S^3$
is homeomorphic to $S^6$, and we will discuss this map explicitly below.
%, showing that it can be chosen to be smooth.
% By the Kervaire-Milnor theorem \cite{KM}, the smooth structure on any topological $S^6$ is unique,
%therefore $S^3\wedge S^3$ is diffeomorphic to $S^6$.
Therefore we get a canonical degree
1 continuous projection map $g:S^3\times S^3 \to S^6$, and we can pullback $\mathsf{G}_2$ via this projection map, giving rise to a
natural principal $\sfSU(3)$-bundle $g^*(\mathsf{G}_2)$ over $S^3\times S^3$. 
Let $g^*(\cE)$ be the associated rank 3 vector bundle over $S^3\times S^3$. Then $\cP=[g^*(\cE)] - [{\bf 1}^2]$ is the Poincar\'e virtual line bundle, 
which is well defined in K-theory and  is the non-trivial generator of $K^0(S^3\times S^3)$ as well as an automorphism of  $K^0(S^3\times S^3)$.
This construction is generalised in section \ref{sect:Tduality} where it is more topological.

% ===================================================================
\subsection{Smashing spheres} \label{smashsez}

To construct a Poincar\'e bundle with connection on $S^3\times S^3$ we will need an explicit formula for the smash 
product map.  In this subsection we will treat the general case $f:S^n\times S^n\to S^n\wedge S^n\cong S^{2n}$.  
The Poincar\'e bundle on $S^n\times S^n$ is constructed by pulling back a vector bundle with minimal nonzero 
Euler class from $S^{2n}$.  In the next two subsections we will restrict our attention to the two examples of interest, 
$n=1$ corresponding to ordinary T-duality and $n=3$ corresponding to spherical T-duality. For a general reference 
to this section, see \cite{Hatcher}.

We begin by recalling that $S^n$ is an $S^{n-1}$ fibration over the interval $I$ which degenerates to a point at the 
two endpoints $\{0,1\}\in I$.  For each point $x_i$ in the $i$th copy of $S^n$, where $i=1$ or $2$, let $r_i\in I$ and 
${\mathbf{v}}_i\in S^{n-1}\subset\RR^n$ be the associated points in $I$ and the unit sphere $S^{n-1}\subset\RR^n$.  
Note that when $r_i=0$ and $1$, all values of ${\mathbf{v}}_i$ are equivalent.  To write the map $f$, it will be convenient 
to embed $S^{2n}$ as the unit sphere in $\RR^{2n+1}$.  The function $f$ can therefore be decomposed into $2n+1$ 
functions $f_i:S^n\times S^n\to \RR$ representing the coordinates in $\RR^{2n+1}$.

We will also decompose $S^{2n+1}$ into an $S^{2n}$ fibration over the interval, where the interval will 
be correspond to the last coordinate in $\RR^{2n+1}$.  We assert furthermore that the $n$-vectors 
$(f_1,...,f_{n})$ and $(f_{n+1},...,f_{2n})$ are parallel to ${\mathbf{v}}_1$ and ${\mathbf{v}}_2$ respectively.  
More precisely, we impose
\begin{equation}
(f_1,...,f_{n})=\alpha_1(r_1,r_2) {\mathbf{v}}_1\,,\hspace{0.4cm} (f_{n+1},...,f_{2n})=\alpha_2(r_1,r_2) {\mathbf{v}}_2\,, \label{sferdecomp}
\end{equation}
where the $\alpha_i$ are nonnegative functions on $I\times I$.  Similarly we demand that $f_{2n+1}$ 
be independent of ${\mathbf{v}}_i$ and so we will write simply $f_{2n+1}(r_1,r_2)$ as a function $I\times I\to [-1,1]$.  
The smash product map $f$ is therefore defined by the three functions $f_{2n+1},\ \alpha_1$ and $\alpha_2$ on $I\times I$.

By the definition of the smash product, $f(S^n\vee S^n)$ is a single point, let it be $({\mathbf{0}}^{2n},-1)$.  
Choose the decomposition of $S^n$ such that $S^n\vee S^n$ is the subset of $S^n\times S^n$ such that $r_1 r_2=0$.  
Then we learn that
\[
f_{2n+1}(0,r_2)=f_{2n+1}(r_1,0)=-1,\hspace{0.4cm} \alpha_i(r_1,0)=\alpha_i(0,r_2)=0\,.
\]
As we would like the smash product $f$ to be smooth, we define
\begin{equation}
f_{2n+1}(r_1,r_2)=-1+r_1 r_2 \tilde{f}(r_1,r_2),\hspace{0.4cm} \alpha_i(r_1,r_2)=r_1r_2\tilde{\alpha}_i(r_1,r_2).\label{fdef}
\end{equation}

The smash product map $f$ must also have degree 1.  For this it is sufficient that the preimage 
of $({\mathbf{0}}^{2n},1)$ contain a single point, which we will fix to be $(r_1,r_2)=(1,1)$.  For this purpose it is sufficient to fix
\[
\tilde{f}(1,1)=2 \,,
\]
and to demand that $\tilde{f}$ be everywhere nondecreasing in both $r_1$ and $r_2$.

Next, recall that all values of ${\mathbf{v}}_i$ are equivalent when $r_i=0$ and $1$.  Therefore $f$ must be 
independent of ${\mathbf{v}}_i$ when $r_i=0$ and $1$.  When $r_i=0$ this condition is satisfied, as the image 
is just $({\mathbf{0}}^{2n},-1)$.  What about $r_i=1$?  Recall that only $(f_1,...,f_{n})$ depends upon 
${\mathbf{v}}_1$ and $(f_{n+1},...,f_{2n})$ upon ${\mathbf{v}}_2$, as they are parallel.  
Therefore a necessary and sufficient condition is that each $n$-vector vanishes when the corresponding $r_i=1$.  In other words, we must impose
\begin{equation}
\tilde{\alpha}_1(1,r_2)=\tilde{\alpha}_2(r_1,1)=0\,. \label{condeq}
\end{equation}

Finally, we must impose that the image of $f$ is actually on the unit sphere
\begin{equation}
1=\sum_{i=1}^{2n+1} f_i^2=f_{2n+1}^2(r_1,r_2)+\alpha_1^2(r_1,r_2)+\alpha_2^2(r_1,r_2) \,, \label{sfera}
\end{equation}
and so
\[
\tilde{f}^2+\tilde{\alpha}^2_1+\tilde{\alpha}^2_2-\frac{2\tilde{f}}{r_1r_2}=0\,.
\]
Now we are done, any triplet $(\tilde{f},\tilde{\alpha}_1,\tilde{\alpha}_2)$ of functions on $I\times I$ 
satisfying the above conditions will induce the smash product $S^n\times S^n\to S^{2n}$.

\subsection{Brief review of 7-twisted K-theory and some calculations}\label{calc}

To define K-theory on the closed, oriented 7D manifold $P$, twisted by a closed 7-form $H$ representing $k$ times the 
generator of $H^7(P,\ZZ)$, we first recall from Corollary 4.7 in \cite{DP} that the generator of 
$H^7(S^7,\ZZ)$ corresponds to the (higher) Dixmier-Douady invariant of an algebra bundle 
$\cE\to S^7$ with fibre a stabilized infinite Cuntz $C^*$-algebra $O_{\infty} \otimes \cK$. 
Now let $f:P \to S^7$ be a degree $k$ continuous map, then $f^*(\cE) \to P$ is an algebra 
bundle with fibre a stabilized infinite Cuntz $C^*$-algebra $O_{\infty} \otimes \cK$ and 
Dixmier-Douady invariant equal to $k$ times the generator of $H^7(P,\ZZ)$. 
Then, by \cite{DP2}, the 7-twisted K-theory is defined as $K^*(P, H) = K_*(C_0(P, f^*(\cE)))$, 
where $C_0(P, f^*(\cE))$ denotes continuous sections of $f^*(\cE)$ vanishing at infinity. 
This shows that $K^*(P, H)$ is well defined, although we will not use the explicit construction.

Our goal here is to compute the 7-twisted K-theory of the total space of a not necessarily principal  $SU(2)$-bundle $P$ with 7-flux $H$ 
over a 
compact simply-connected, connected, oriented 4D manifold $M$. 
The strategy of proof is as follows. We first compute the K-theory of $M$ and then use it to compute 
the K-theory of $P$ using the Gysin sequence in K-theory (cf. \cite{EM}). Using the computation of the 
K-theory of $P$, it is then an easy matter to compute the 7-twisted K-theory of $P$. 

First note that by our hypotheses on $M$ that $H^2(M, \ZZ) \cong \ZZ^{b_2}$ is torsion free, where $b_2=b_2(M)$ is the 2nd Betti number of $M$.  Also $H^0(M, \ZZ)\cong \ZZ\cong H^4(M, \ZZ)$, 
and finally that $H^1(M, \ZZ) = 0 = H^3(M, \ZZ)$.

Consider the 4D stage of the Postnikov tower for $BSU$, denoted by $BSU^{(4)}$. Then $\pi_i(BSU^{(4)}) = \pi_i(BSU)$ for $i\le 4$ and 
$\pi_i(BSU^{(4)})=0$ for $i>4$. That is, $BSU^{(4)}=K(\ZZ,4)$ and  the natural map 
$[M, BSU] \cong [M, BSU^{(4)}] = [M, K(\ZZ, 4)] = H^4(M, \ZZ)\cong \ZZ$ and the isomorphism is given by the 2nd Chern class.

The exact sequence $1\to SU\to U \stackrel{det}{\to} U(1)\to 1$ gives rise to a fibration $BSU\to BU \stackrel{Bdet}{\to} BU(1)$,  and to an exact sequence 
$0\to [M, BSU] \to [M, BU]\stackrel{Bdet}{\to} [M, BU(1)]\to 0$. That is, $0\to H^4(M, \ZZ)\to \tilde K^0(M) \to H^2(M, \ZZ)\to 0$, where $\tilde K^0$ denotes the reduced K-theory. But this 
last sequence splits, as line bundles on $M$ naturally define elements in $ \tilde K^0(M)$. Therefore,
$$
K^0(M) \cong \ZZ \oplus H^2(M, \ZZ) \oplus \ZZ,
$$
and is consistent  with the computation using the Atiyah-Hirzebruch spectral sequence.
Similar arguments show that $K^1(M)=0$, since the odd dimensional cohomology of $M$ vanishes, and completes the computation of the K-theory of $M$.

Next we use this computation and the Gysin sequence (cf. \cite{EM}) to compute the K-theory of $P$, where $S^3\to P\stackrel{\pi}{\to} M$ is the relevant 
fibre bundle, which becomes in this case,
$$
0\to K^1(P) \stackrel{\pi_{*}}{\longrightarrow} K^0(M) \stackrel{e\cup}{\longrightarrow} K^0(M) \stackrel{\pi^{*}}{\longrightarrow}  K^0(P) \to 0.
$$
Here $e\cup \xi = \chi(P)\, rank(\xi) \,pt_!$, where $ \chi(P)= c_2(P)$ is the Euler characteristic or the 2nd Chern number and $pt:\{x\}\hookrightarrow M$ is the inclusion of a point, and $pt_!$ the associated Gysin map, $pt_! \in K^0(M)$. We deduce that 
$$
K^1(P) \cong \ZZ^{b_2+1} 
$$
if $\chi=\chi(P) \ne 0$ and $K^1(P) = \ZZ^{b_2+2}$ if $\chi=0$, where $b_2=b_2(M)$ is the 2nd Betti number of $M$. Also 
$$
K^0(P) \cong \ZZ^{b_2+1} \oplus \ZZ/\chi \ZZ
$$
if $\chi=\chi(P) \ne 0$ and $K^0(P) = \ZZ^{b_2+2}$ if $\chi=0$. 

Now let $pt:\{x\}\hookrightarrow P$ be the inclusion of a point, and $pt_!$ the associated Gysin map, $pt_! \in K^1(P)$. Then we can view
the 7-flux $H$ as an element in $K^1(P)$ as follows. Firstly, $H \in H^7(P, \ZZ)\cong \ZZ$ can be identified as an integer $h\in \ZZ$.
Then let $H \cup \xi = h \, rank(\xi)\, pt_!$, and $K^*(P, H)$ is the cohomology of the complex,
$$
0\to K^0(P) \stackrel{H\cup}{\longrightarrow} K^1(P) \to 0
$$
Therefore 
$$
K^0(P, H) \cong \ZZ^{b_2} \oplus \ZZ/\chi \ZZ,\qquad \text{if} \quad h\ne 0\, \&\, \chi\ne 0
$$
and 
$$
K^1(P, H) \cong \ZZ^{b_2} \oplus \ZZ/h\ZZ, \qquad \text{if} \quad h\ne 0\, \&\, \chi\ne 0.
$$
In the case when $\chi=0$ and $h\ne 0$,
$$
K^0(P, H) \cong \ZZ^{b_2+1},  \qquad \text{if} \quad \,  h\ne 0 \&\, \chi= 0
$$
and 
$$
K^1(P, H) \cong \ZZ^{b_2+1}, \qquad \text{if} \quad h\ne 0\, \&\, \chi= 0.
$$

Now recall from the discussion in \cite{BEM142} that oriented $S^3$ bundles over $M$ as above are determined by fixing the first Pontryagin class, the second Stiefel Whitney class and the Euler class (or 2nd Chern class).
Spherical T-duality asserts that upon fixing the first Pontryagin class $p_1(P)=p_1(\hat P)$ (where $\hat P$ denotes the spherical T-dual $S^3$ bundle over $M$) and the second Stiefel-Whitney class $w_2(P)=w_2(\hat P)$, then there is an exchange $\hat h=\chi$ and $\hat \chi=h$, where $\hat \chi = \chi(\hat P)$ is the Euler class of $\hat P$ and 
the spherical T-dual 7-flux $\hat H = \hat h\, \text{vol}_{\hat P}$. Therefore $K^0(P, H)  \cong K^1(\hat P, \hat H) $ and $K^1(P, H)  \cong K^0(\hat P, \hat H) $. In the next section, we will argue that this isomorphism is determined by the Poincar\'e 
virtual line bundle $\cP$ defined earlier in the section.

\section{Spherical T-duality Induces an Isomorphism on Higher Twisted K-Theory}\label{sect:Tduality}

The goal of this section is to show that the Poincar\'e virtual line bundle described in the previous section,  induces the spherical T-duality isomorphism on 7-twisted K-theories on the total spaces of $\sfSU(2)$-bundles (with 7-flux) over 4D manifolds  that are spherically T-dual. This makes explicit the isomorphisms in Section \ref{calc}. Our proof is modelled on 
that of \cite{BS} (see also \cite{LSW}). Some of the arguments generalise to higher dimensions, as indicated.

\subsection{Poincar\'e Bundles on General Spheres}

Toplogically, it is not difficult to extend the above construction of the Poincar\'e bundle to arbitrary dimensional spheres.  The smash product
$$
f:S^n\times S^n\to S^{2n}
$$  
can be used to pull back the principal ${\sfSU(n)}$ bundle $P\rightarrow S^{2n}$ whose clutching map generates $\pi_{2n-1}(\sfSU(n))=\ZZ$.  The virtual line bundle equal to $f^*P$ minus the rank $n-1$ trivial bundle is the Poincar\'e virtual line bundle.  We will restrict our attention to $n>0$.

\subsection{Thom class}

In this section, we will adapt the proof of \cite{BS}, generalized to higher dimensional sphere bundles following \cite{LSW}, to show that spherical T-duality yields an isomorphism of twisted K-theory on orientable $S^n$ bundles $P$ over $(n+1)$-dimensional orientable manifolds $M$, where the K-theory is twisted by an element $H\in H^{2n+1}(P)\cong \ZZ$.  The case treated in \cite{BS} corresponds to $n=1$.  In that case, given the pair $(P,H)$ the T-dual is unique.  More generally, the T-dual is not unique and so we will show that the twisted K-theories of all T-duals are isomorphic.  Uniqueness does reappear in several cases, such as principal $n=3$ bundles \cite{BEM14} or even nonprincipal $n=3$ bundles if the  Pontryagin class and the second Stiefel-Whitney class are held fixed.

The novelty in the treatment of T-duality in Ref.~\cite{BS} is the introduction of a sphere bundle $S(V)$ which is fiberwise the join of $P$ and $\hat{P}$.  In the special cases $n=1$ and $n=3$, the fibers of $P$ and $\hat{P}$ are the Lie groups $\sfU(1)$ and $\sfSU(2)$.  If the bundle is principal then $S(V)$ can also be constructed as the $S^{2n+1}$ subbundle of $V\rightarrow M$, defined to be the direct sum of the $\CC^{(n+1)/2}$ bundles associated to $P$ and $\hat{P}$.   More generally, following \cite{LSW}, the join construction may be realized in the spirit of \cite{BS} by embedding $P$ and $\hat{P}$ in $\RR^{n+1}$ bundles $E$ and $\hat{E}$, whose direct sum again yields $V$.  

The Thom class is an element $Th\in H^{2n+1}(S(V),\ZZ)=\ZZ$.    Distinct Thom classes are related by pullbacks of classes on $M$.  However, in our case $M$ is an $(n+1)$-manifold and so $H^{2n+1}(M)=0$.  Therefore in the case at hand, the Thom class is unique and in fact, as its integral on $S(V)$ generates $H^{n+1}(M)=\ZZ$, the Thom class must represent the element
$$
[Th]=1\in H^{2n+1}(S(V),\ZZ)=\ZZ.
$$
Similarly, although the Gysin sequence can be used to show that the cup product of the Euler classes of $P$ and $\hat{P}$ vanishes \cite{BEM14}, the Euler class of $S(V)$ lies in $H^{2n+2}(M)=0$ and so again vanishes, and so the Thom class exists.

The correspondence space $P\times_M \hat{P}$ may be constructed as in the cases $n=1$ \cite{BEM} and $n=3$ \cite{BEM14} and, as was done there, we impose that
\begin{equation}
p^*\hat{H}=\hat{p}^* H \label{hiso}
\end{equation}
at the level of cohomology.   Ref.~\cite{BS} introduces another characterization of $H$ and $\hat{H}$ as
\begin{equation}
H=i^* Th,\hspace{.4cm} \hat{H}=\hat{i}^*Th \label{hdef}
\end{equation}
where $i:P\rightarrow S(V)$ and $\hat{i}:\hat{P}\rightarrow S(V)$ are the inclusions of the $S^n$ bundles in $S(V)$, described fiberwise by the join operation or the sphere subbundle of $E\oplus\hat{E}$. 

In fact the definition (\ref{hdef}) implies (\ref{hiso}), as is shown in the case $n=1$ in Lemma 2.13 of \cite{BS} and more generally in \cite{LSW}.  The proof for general $n$ proceeds identically.  First, note that the $S^n$ and $\hat{S}^n$ in the join construction $S^{2n+1}=S^n*S^n$ are homotopic.   This homotopy yields a homotopy from $\hat{i}\circ p:P\times_M\hat{P}\rightarrow S(V)$ to  $i\circ \hat{p}:P\times_M\hat{P}\rightarrow S(V)$, and so as in the case $n=1$
\[
p^* \hat{H} = p^* \hat{i}^* Th=\hat{p}^*i^*Th=\hat{p}^* H .
\]

\subsection{7-twisted K-theory}

To adapt the proof of Ref.~\cite{BS} that T-duality is an isomorphism twisted K-theory to the case of $S^n$ bundles and higher twisted K-theory, we will need to use the fact that (2n+1)-twisted K-theory on an oriented (2n+1)-manifold\footnote{Note that this is {\it{not}} the most general twisted K-theory on a (2n+1)-manifold, we do not expect that our conclusions can be generalized to other twists.}  indeed is a twisted cohomology theory, and in particular that it satisfies several key properties.   As twists correspond to elements of $H^{2n+1}$, they correspond to maps $f$ to $K(\ZZ,2n+1)$.  Automorphisms of these are given by maps to the free loopspace $LK(\ZZ,2n+1)$ such that a given map $f$ is fixed.  Homotopy classes of automorphisms of (2n+1)-twists therefore correspond to maps to $K(\ZZ,2n)$, or equivalently $2n$-cohomology classes over the integers.  

The third paragraph of 3.1.5 of \cite{BS} describes a special case of the relation between these automorphisms and the twists themselves.  This is essentially a higher gerbe generalization of the clutching construction.  In the case $n=1$, it is the statement that a 1-gerbe on $X=A\cup B$ can be created by gluing together trivial gerbes on $A$ and $B$ with the transition function given by a line bundle on $A\cap B$.   This gives an isomorphism between $H^3(X,\ZZ)$ and $H^2(A\cap B)$, corresponding to an isomorphism between 3-classes of gerbes on $X$ and Chern classes of line bundles on $A\cap B$.    A similar construction applies to maps $X\rightarrow K(\ZZ,2n+1)$.  These can be trivialized on $A$ and $B$ but glued together using a transition map $A\cap B\rightarrow K(\ZZ,2n)$, yielding the desired isomorphism.  It would be interested to have a higher gerbe interpretation of this construction using the formalism of Ref.~\cite{LSW}, whose $\Lambda$ operation is an example of such an automorphism.

The application to suspensions in \cite{BS} is essentially an example in which $X$ is homotopic to $\Sigma(A\cap B)$.  The isomorphism is guaranteed to exist in this case as suspensions are homotopically the inverse of the based loop spaces considered above.

As is described, in the case $n=1$, in paragraph 3.1.5 of \cite{BS}, homotopies $h:\RR\times Y\rightarrow X$ uniquely induce such automorphisms $u(h)\in H^{2n}(Y,\ZZ)$.   The automorphism again is just the clutching map which generates the gerbe corresponding to the twist.   

The twist of a higher twisted K-theory on a $(2n+1)$-manifold $X$ is entirely characterized by a cohomology class in $H^{2n+1}(X)=\ZZ$.   Higher twisted K-theory satisfies all of the usual axioms of a twisted cohomology theory, including crucially the Mayer-Vietoris property as was shown in Theorem 2.7 of Ref.~\cite{DP2}.

\subsection{Spherical T-admissibility}
The critical step in the proof of the K-theory isomorphism in Ref.~\cite{BS} is the demonstration that twisted K-theory is $T$-admissible, or in other words that T-duality yields an isomorphism of K-theory when $M$ is a point.   Needless to say, as the sphere bundle over a point is $n$-dimensional this twist will necessarily be trivial for twisted K-theory and for higher twisted K-theory.

The proof again uses the homotopy $h:I\times S^n\times \hat{S}^n\rightarrow S^{2n+1}$ from $i:S^n\rightarrow S^{2n+1}$ to $\hat{i}:\hat{S}^n\rightarrow S^{2n+1}$ given by the join construction.  The T-duality map is defined to be
$$
T=p_! u(h)^* \hat{p}^*:K(S^n,H)\rightarrow K(\hat{S}^n,\hat{H}).
$$ 
What is $u(h)$?  Recall that $u(h)$ is the automorphism of twisted K-theory which serves as the clutching function in the construction of the Thom class.  In the case at hand, $S(V)$ is just $S^{2n+1}$ and the Thom class is just its top class.  The $S^{2n+1}$ can be constructed as $S^n * \hat{S}^n$, in other words as an $S^n\times S^n$ fibration over an interval in which one $S^n$ degenerates at each end.   Fixing a point $x$ on the interior of the interval, $S^{2n+1}$ can be decomposed into $A$ and $B$ consisting of the fibers over the left and the right of the point, including the fiber over $x$ itself.   Now the clutching function is just the map from the $S^n\times \hat{S}^n$ over $x$ to $K(\ZZ,2n)$.  

Which  map is it?  The Thom class integrated over $S^{2n+1}$ gives the generator of the cohomology of a point, and so it generates $H^{2n+1}(S^{2n+1})=\ZZ$.  Therefore the clutching function generates $H^{2n}(S^{2n})=\ZZ$.  The corresponding automorphism on K-theory can be represented by the minimal bundle on $S^n\times\hat{S}^n$.  As an automorphism, it must have $c_0=1$.  Also $c_n=(n-1)!$, which is the minimal nonzero value.  All other classes vanish, indeed there is no cohomology available to support them.  No complex vector bundles has these properties, however there is a virtual bundle with these properties.  It is just our Poincar\'e virtual line bundle.

The demonstration of T-admissibility then follows from a straightforward calculation.   The action of $u(h)$ is just a tensor product of bundles, which multiplies Chern classes.  For concreteness, let $n$ be odd.  Then if $1$ is the generator of $K^0(S^n)$ and $u$ generates $K^1(\hat{S}^n)$ then 
$$
T(1)=B^{(n+1)/2}(u),\hspace{.4cm}
T(u)=B^{(n-1)/2}(1)
$$
where $B:K^n\rightarrow K^{n-2}$ is the Bott periodicity element.  This is a degree shifting isomorphism.   If $n$ is even and $1$ and $u$ are generators of $K^0(S^n)=\ZZ^2$ corresponding to the trivial and rank zero non trivial virtual bundles, while $K^1(S^n)=1$, then
$$
T(1)=B^{n/2}(u),\hspace{.4cm}
T(u)=B^{n/2}(1)
$$
which again is an isomorphism, this time with no degree shift.

Note that, unlike the case $n=1$, in general the Chern class does not give an isomorphism between $K^0(S^n\times\hat{S}^n)$ and $H^{2n}(S^n\times\hat{S}^n)=\ZZ$.  Rather the Chern class of the generator of K-theory corresponds to the element
$$
c_n=(n-1)!\in H^{2n}(S^n\times\hat{S}^n)=\ZZ.
$$
On the other hand, the loop space argument above indicates that the automorphisms of the twist are the entire cohomology group.   Of course these are related, the class of the automorphism is just the Chern class divided by $(n-1)!$.    %It would be interesting to understand how the isomorphisms of the twists, represented by $\ZZ$, and the isomorphisms of the K-theory, represented by $=\left(\frac{n+1}{2}\right)!\ZZ$, may be related, generalizing condition (3) of 3.1.2 in Ref.~\cite{BS}.

\subsection{Twisted K-theory}

The final step of the proof considers a general $(n+1)$-dimensional oriented base $M$.  One uses (\ref{hiso}) to generate an automorphism $u$ which plays the role of $u(h)$ above

\begin{equation*}
\xymatrix{  & P\times_M \widehat P \ar[dl]_{\widehat p} \ar[dr]^{{p}} &  \\
P \ar[dr]_\pi  && \widehat{P}  \ar[dl]^{\widehat \pi} \\
& M &  }
\end{equation*}

This is the parameterized version of the situation considered earlier. In particular, we have a homotopy
$h:I\times P\times_M\widehat P\rightarrow S(V)$ from $i\circ \hat{p}$ to $\hat i\circ p$.
It induces the morphism
$$
u:p^* \widehat\cH=p^*\hat i^*\cK\stackrel{}{\cong} (\hat i\circ p)^*\cK\stackrel{u(h)}{\cong}
(i\circ \hat p)^*\cK\stackrel{}{\cong}
\hat p^*i^*\cK=\hat p^*\cH \,, 
$$
which is natural under pullback of bundles.  Here $u(h)$ fiberwise is again the twisted automorphism corresponding to the generator of the automorphism group $H^{2n}$.  It induces a map $u^*$ on K-theory given by a tensor product with the Poincar\'e virtual line bundle.

We define the spherical T-duality transformation on $(2n+1)$-twisted K-theory on $(2n+1)$-dimensional manifolds as
$$T:=p_!\circ u^*\circ \hat p^*:K(P,\cH)\rightarrow K(\widehat P,\widehat
\cH)\ .$$

The main theorem of the present section is the following.
Assume that $M$ is  homotopy equivalent to a finite complex.
\begin{theorem}\label{thm:spherical-T-duality}
The spherical T-duality transformation T is an isomorphism.
\end{theorem}
\begin{proof}The twisted K-theory isomorphism follows from the Mayer-Vietoris property of twisted K-theory together with two lemmas proved in \cite{BS} demonstrating that pullbacks and the Mayer-Vietoris maps both commute with T-duality.  The proofs of these statements, using the higher-dimensional definitions above, are identical in our case and the isomorphism of higher twisted K-theories follows.  When $n$ is odd, the degrees are shifted as was seen in the K-admissibility proof.  This is a consequence of the fact that, in the T-duality map, only the pushforward changes the degree and it shifts the degree by $n$.
\end{proof}

% ===================================================================
\subsection{The spherical T-duality group}

We now restrict our attention to $n=3$.

Consider $\sfSU(2)$ as the unit quaternions ie $\mathsf{Sp}(1)$. Then quaternionic conjugation
is an orientation reversing automorphism of $\sfSU(2)$. So given a principal $\sfSU(2)$-bundle
$P$ over a 4-dimensional manifold $M$, let $x.g$ denote the right action of
$g\in \sfSU(2)$ on $x\in P$. Then $x.{\bar g}$ also gives a right action of
$\sfSU(2)$ on $P$, where ${\bar g}$  is the quaternionic conjugate of $g$. It is again a free
action, so it defines a principal $\sfSU(2)$-bundle with the same total space and with
2nd Chern class the negative of 
$c_2(P)$. This gives the action of the non-trivial element $-1 \in \mathsf{GL}(1, \ZZ)$ on
spherical T-dualities (with 4D base). Now $-1\in  \mathsf{GL}(1, \ZZ)$ corresponds to the element
$(-1,-1) \in \mathsf O(1,1,\ZZ)$ via the canonical embedding of $\mathsf{GL}(1, \ZZ)$  in $\mathsf O(1,1,\ZZ)$. The other generator
of $\mathsf O(1,1,\ZZ)$ is the $2\times 2$ matrix with $1$'s on the off-diagonal and $0$'s on the diagonal.
This element exchanges the 2nd Chern class and the 7-flux i.e.\ is the spherical T-duality element.
Therefore $\mathsf O(1,1,\ZZ)$  is the spherical T-duality group.\\

% ===================================================================
\section{The putative relation to SUGRA and M-theory}

In this section, we suggest the relevance of spherical T-duality to 11 dimensional supergravity and M-theory.

Recall that the action of eleven dimensional supergravity is (cf. \cite{CJS}),
$$\label{eleven}
I_{11}=I_{grav} + I_{G_4} + I_{C.S.}+ I_{fermi} + I_{coupling}
$$
where
\begin{eqnarray}
I_{grav} &= &\frac{1}{2 \kappa_{11}^2} \int_{Y}  {\hat{\cR}} 
\;d {\rm vol} \\
I_{G_4} &= & - \frac{1}{2 \kappa_{11}^2} \frac{1}{2 \cdot 4!} \int_{Y}
| G_4 |^2\; d {\rm vol}\\
I_{C.S.} &=& -\frac{1}{12 \kappa_{11}^2} \int_{Y} C_3 \wedge G_4 \wedge
G_4
\\
I_{fermi} & =& \frac{1}{2 \kappa_{11}^2} \frac{1}{2} \int_{Y}
\bar{\psi}
D_{R-S} \psi \; d {\rm vol}
\end{eqnarray}
Here $Y$ is 11-dimensional spacetime, $d {\rm vol}= d^{11}x \sqrt{-g}$, $\hat{\cR}$ is the scalar
curvature of $Y$, $G_4$ is the
four-form
field
strength, which, when cohomologically trivial, is equal to $dC_3$. The
fermions involve the kinetic action of $\psi$ involving
the Rarita-Schwinger operator $D_{R-S}$.  One can view $D_{R-S}$
as 
%\cite{DMW1, DMW2} 
the Dirac operator coupled to the vector bundle
associated to the
virtual bundle $TY - 3{\O}$, where the ${\O}$ factors correspond to
subtraction of ghosts.
$I_{coupling}$ corresponds to coupling of
$\psi$
to $G_4$ as well as quartic $\psi$ self-couplings and we refer the reader to \cite{CJS} for details.

What concerns us here are the equations of motion of the source-free Bianchi identity, which are,
\begin{equation}\label{bianchi}
dG_4=0, \qquad d(\star G_4) = -\frac{1}{2} G_4\wedge G_4.
\end{equation}
These can be modified by adding sources, namely the membrane $M2$
and the fivebrane $M5$, respectively.
Consider the   {\em Freund-Rubin solutions} to these equations, 
$$
Y= X^4 \times W^7, \qquad G_4= t\, \text{vol}_X
$$
where $X^4={\rm Ad}S_4= O(2,3)/O(1,3)$ is anti-deSitter space and $W^7$ is a 7-dimensional Einstein manifold, and $t\in\RR$ is a constant.
This is called the near horizon geometry of a stack of M2-branes.
Notice that since $G_4\wedge G_4=0$, therefore $d(\star G_4) = 0$ and the source-free Bianchi equations \eqref{bianchi} above are satisfied. 

\begin{table}
\centering
\begin{tabular}{|c|c|c|}
\hline\hline
& String Theory & M-Theory /11D SUGRA \\ 
& $X^4 \times E^6$ & $X^4 \times W^7$ \\ \hline
& Complex manifold & Contact manifold \\
$\mathcal N=1$ & K\"ahler  & Sasakian   \\
$\mathcal N=2$ &Calabi-Yau  & Sasaki-Einstein  \\
$\mathcal N=3$ & Hyper-K\"ahler  & 3-Sasakian\\ \hline\hline
&$S^1$ & $S^3$ \\
& Strings & M2- and M5-branes \\
& $H\in \mathrm{H}^3(E,\ZZ)$ & $H\in \mathrm{H}^7(W,\ZZ)$ \\
& Mirror Symmetry / T-duality & Spherical T-duality? \\ 
 \hline\hline 
& $\xymatrix{ S^1 \ar[r] & E \ar[d] \\ & M }$ & 
$\xymatrix{ S^3 \ar[r] & W^7 \ar[d] \\ & \CC P^2 }$\\ \hline\hline
\end{tabular} 
\end{table}

Setting $H_7=\star G_4$, we get a closed 7-form $H_7$ on $W$. Taking $W$ to be 
an Aloff-Wallach space as in \cite{BEM142} and applying spherical T-duality, we get a
spherical T-dual Aloff-Wallach space $\widehat W$ and spherical T-dual flux $\widehat H_7$.

Then $$\widehat Y = X^4 \times \widehat W^7, \qquad \widehat G_4=\star \widehat H_7$$
also satisfy the source-free Bianchi equations \eqref{bianchi} above, and therefore is a topological 
distinct M2-brane near horizon geometry.  A naive analogy with T-duality would then suggest an equivalent between (some subsector of) these two topologically distinct M-theory compactifications.

\appendix
% ===================================================================
\section{On the geometry of Poincar\'e bundles}

As an illustration of our current construction, we give another construction of the Poincar\'e line bundle with 
connection next. Consider $\sfG=\sfSU(2)$.  Using the parametrization
\begin{equation*}
g = e^{i \phi \sigma^3/2} e^{i \theta\sigma^1/2}  e^{i \psi \sigma^3/2},\hspace{.4cm}
\phi\in[0,2\pi),\hspace{.4cm}
\theta\in[0,\pi),\hspace{.4cm}
\psi\in[0,4\pi) \,,
\end{equation*}
the Maurer-Cartan form for $\sfG$ can be written as
\begin{equation*}
\omega = g^{-1}dg = \sum_i e^i \left( \frac{i \sigma^i}{2}  \right)\,,
\end{equation*}
where, in particular,
\begin{equation*}
e^3 = d\psi + \cos\theta \, d\phi \,.
\end{equation*}
We can use $A = e^3/2$ as a principal connection on the principal $\sfU(1)$-bundle
$S^3$ over $S^2$, where the normalization is chosen such that the integral of $A$ over the fiber is equal to one.  Then
\begin{equation*}
F = dA = -\frac{\sin\theta}{2} \, d\theta \wedge d\phi\,,
\end{equation*}
and
\begin{equation*}
c_1 = \frac{1}{2\pi} \int_{S^2} F = -1 \,.
\end{equation*}
%(sign depending on orientation of $S^2$).

To obtain the Poincar\'e bundle on $S^1\times S^1$, we need to pull this bundle back by the smash product 
$f : S^1 \times S^1 \to S^1 \wedge S^1 \cong S^2$.  This is the case $n=1$ of the general construction treated 
in Sec.~\ref{sect:poincare}.  If $\beta\in[0,2\pi]$ and $\gamma\in[0,2\pi]$ are the coordinates for the two copies of 
$S^1$, then we can define the two intervals by the maps
\[
r_1:S^1\to I:\theta\mapsto {\rm{sin}}\left(\frac{\beta}{2}\right),\hspace{.4cm}
r_2:S^1\to I:\phi\mapsto {\rm{sin}}\left(\frac{\gamma}{2}\right) \,.
\]
Fibered over each interval is an $S^0$ with coordinates ${\mathbf{v}}_i\in\{-1,1\}$.

The conditions~(\ref{condeq}) are that when $r_1=1$, corresponding to $\theta=\pi$, $\tilde{\alpha}_1=0$ 
and also when  $r_2=1$, corresponding to $\phi=\pi$, $\tilde{\alpha}_2=0$.
We satisfy these conditions by choosing
\begin{equation}
\tilde{\alpha}_1=2\left|{\mathrm{cos}}\left(\frac{\beta}{2}\right)\right|\,,\hspace{.4cm}
\tilde{\alpha}_2=2{\mathrm{sin}}\left(\frac{\beta}{2}\right)\left|{\mathrm{cos}}\left(\frac{\gamma}{2}\right)\right|\,. \label{toro}
\end{equation}
Note that the absolute values are multiplied by elements ${\mathbf{v}}_i=\pm 1\in S^0$ in Eqn.~(\ref{sferdecomp}).  
The effect of this multiplication is simply to remove the absolute values, resulting in a smooth map.
Inserting this into Eq.~(\ref{fdef}) and imposing (\ref{sfera}) we obtain the smash product map
\[
f(\beta,\gamma)=\left({\rm{sin}}\left(\beta\right){\rm{sin}}\left(\frac{\gamma}{2}\right),
{\rm{sin}}^2\left(\frac{\beta}{2}\right){\rm{sin}}\left(\gamma\right),-1+2{\rm{sin}}^2\left(\frac{\beta}{2}\right)
{\rm{sin}}^2\left(\frac{\gamma}{2}\right)\right) .
\]

In terms of spherical coordinates on the $S^2$ this map is
\begin{eqnarray}
(\theta,\phi)&=&\left({\rm{arccos}}(-z),{\rm{arctan}}\left(\frac{y}{x}\right)\right)\nonumber\\
&=&\left({\rm{arccos}}\left(1-2\rm{sin}^2\left(\frac{\beta}{2}\right)\rm{sin}^2
\left(\frac{\gamma}{2}\right)\right),\rm{arctan}\left(\rm{tan}\left(\frac{\beta}{2}\right)\rm{cos}\left(\frac{\gamma}{2}\right)\right) 
\right) \,. \nonumber
\end{eqnarray}
To pullback the curvature we will need the derivatives of this map
\begin{eqnarray}
\frac{\partial\theta}{\partial\beta}&=&-\frac{\rm{cos}\left(\frac{\beta}{2}\right)
\rm{sin}\left(\frac{\gamma}{2}\right)}{\sqrt{1-\rm{sin}^2\left(\frac{\beta}{2}\right)\rm{sin}^2\left(\frac{\gamma}{2}\right)}}\,,\hspace{.4cm}
\frac{\partial\theta}{\partial\gamma}=-\frac{\rm{sin}\left(\frac{\beta}{2}\right)
\rm{cos}\left(\frac{\gamma}{2}\right)}{\sqrt{1-\rm{sin}^2\left(\frac{\beta}{2}\right)\rm{sin}^2\left(\frac{\gamma}{2}\right)}}\,,\nonumber\\
\frac{\partial\phi}{\partial\beta}&=&\frac{\rm{cos}\left(\frac{\gamma}{2}\right)}{2\left(1-
\rm{sin}^2\left(\frac{\beta}{2}\right)\rm{sin}^2\left(\frac{\gamma}{2}\right)\right)}\,,\hspace{.4cm}
\frac{\partial\phi}{\partial\gamma}=-\frac{\rm{sin}\left(\frac{\beta}{2}\right)\rm{cos}\left(\frac{\beta}{2}\right)
\rm{sin}\left(\frac{\gamma}{2}\right)}{2\left(1-\rm{sin}^2\left(\frac{\beta}{2}\right)\rm{sin}^2\left(\frac{\gamma}{2}\right)\right)}\nonumber \,.
\end{eqnarray}

Finally we can compute the curvature on the Poincar\'e bundle as the pullback of the curvature on the Hopf bundle
\begin{equation*}
f^* F = - \frac{\rm{sin}(\theta)}{2}\left(\frac{\partial\theta}{\partial\beta}\frac{\partial\phi}{\partial\gamma}-
\frac{\partial\theta}{\partial\beta}\frac{\partial\phi}{\partial\gamma}\right)d\beta\wedge d\gamma=-
\frac{1}{2}\sin^2\left(\frac{\beta}{2}\right)\cos\left(\frac{\gamma}{2}\right) d\beta\wedge d\gamma .
%f^* F = - d\hat\theta \wedge d\hat\phi + \text{exact term}\,,
\end{equation*}
%i.e.\ the curvature of the Poincar\'e linebundle over $S^1 \times S^1$.
As a consistency check, we can integrate this curvature to obtain the Chern class
\begin{equation*}
c_1 = \frac{1}{2\pi} \int_{T^2} f^*F = -1 \,.
\end{equation*}

%\section
As $S^3$ is the group manifold of $\sfSU(2)$, 
we may use the group structure to re-express the maps used in the general construction above.  
One realization of the decomposition of $S^3$ into an $S^2$ fibration over an interval is the decomposition 
of $\sfSU(2)$ into conjugacy classes corresponding to elements with eigenvalues $e^{\pm i \pi r}$.  
These conjugacy classes are of topology $S^2$ for $r\in (0,1)$ and are points, consisting of the elements 
$\pm\mathbf{1}\in\sfSU(2)$, for $r=\{0,1\}$.  More specifically, for each $g\in\sfSU(2)$ we define $r\in I$ and $v\in S^2\subset\RR^3$ by
\[
g={\rm{exp}}(i r {\mathbf{v}}\cdot {\mathbf{\sigma}}) \,,
\]
where ${\mathbf{\sigma}}$ are the Pauli matrices such that $i{\mathbf{\sigma}}$ generates the Lie algebra $\mathfrak{su}(2)$.  
Using this decomposition, to each point $x\in S^3\times S^3$ we can identify a quadruplet $(r_1,{\mathbf{v}}_1,r_2,{\mathbf{v}}_2)$ 
where all values of ${\mathbf{v}}_i$ are identified when $r_i=0$ or $r_i=1$, as in the general construction in Subsec.~\ref{smashsez}.

To complete the construction, we need to define the pair $\tilde{\alpha}_i$ of functions on $I\times I$.  
The functions $\tilde{\alpha}_i$ can be defined as in Eqn.~(\ref{toro}) in the case $n=1$
\[
\tilde{\alpha}_1=2\sqrt{1-r_1^2}\,,\hspace{.4cm}
\tilde{\alpha}_2=2r_1\sqrt{1-r_2^2} \,.
\]
The third function, $\tilde{f}$, is defined by (\ref{sfera}), choosing the branch which gives a winding number of 1
\[
f_{2n+1}=-1+2r_1^2r_2^2 \,,
\]
as in the case $n=1$, thus completing the construction of the smash product $f:S^3\times S^3\to S^6$
\[
f(r_1,{\mathbf{v}}_1,r_2,{\mathbf{v}}_2)=\left(2r_1r_2\sqrt{1-r_1^2}\ {\mathbf{v}}_1,2r_1^2r_2\sqrt{1-r_2^2}\ {\mathbf{v}}_2,-1+2r_1^2r_2^2\right)\,.
\]
The connection on the Poincar\'e bundle is then $f^* A$.

If we want to calculate this connection explicitly, then we may proceed as in the torus case of the previous subsection.  First we construct an arbitrary element of $G_2$ as
\begin{eqnarray*}
g &=& e^{\left(\pi i
\ {\rm{arccos}}\left(-1+2r_1^2r_2^2\right)
\frac{\sqrt{1-r_1^2}{\mathbf{v}}_1\cdot{\mathbf{a}}_1+r_1\sqrt{1-r_2^2}{\mathbf{v}}_2\cdot{\mathbf{a}}_2}{\sqrt{1-r_1^2r_2^2}}
\right)}
e^{\left(\pi i\left(-c_8M_3+\sum_{i=1}^7c_i F_i
\right)\right)}\,, \\
{\mathbf{a}}_1&=&(M_1,M_2,M_4)\,,\hspace{.4cm}
{\mathbf{a}}_2=(M_5,M_6,M_7) \,,
%{i \phi \sigma^3/2} e^{i \theta\sigma^1/2}  e^{i \psi \sigma^3/2},\hspace{.4cm}
%\phi\in[0,2\pi),\hspace{.4cm}
%\theta\in[0,\pi),\hspace{.4cm}
%\psi\in[0,4\pi)
\end{eqnarray*}
where $F_i$ and $M_i$ are generators of $G_2$ defined in Ref.~\cite{G2}, where it was noted that $F_i$ together 
with $-M_3$ generate an $\sfSU(3)$ subgroup.  

As the Maurer-Cartan form is $\sfSU(3)$-invariant, to obtain the horizontal part of the connection it will be sufficient 
to restrict our attention to $c_i=0$, where $g$ is a section of the bundle $\sfG_2\longrightarrow S^6$ restricted 
to the compliment of the north pole.  As in the toroidal case, it will be convenient to work in spherical coordinates.  Therefore we define
\begin{equation*}
{\bf{v}}_i=({\rm{sin}}(\theta_i){\rm{cos}}(\phi_i),{\rm{sin}}(\theta_i){\rm{sin}}(\phi_i),{\rm{cos}}(\theta_i))\,.
\end{equation*}
Thus we find
\begin{equation*}
g = e^{\left(\pi i
\ {\rm{arccos}}\left(-1+2r_1^2r_2^2\right)
\frac{\sqrt{1-r_1^2}\left({\rm{s}}(\theta_1){\rm{c}}(\phi_1)M_1+{\rm{s}}(\theta_1){\rm{s}}(\phi_1)M_2+
{\rm{c}}(\theta_1)M_4\right)+r_1\sqrt{1-r_2^2}\left({\rm{s}}(\theta_2){\rm{c}}(\phi_2)M_5+{\rm{s}}(\theta_2){\rm{s}}(\phi_2)M_6+{\rm{c}}(\theta_2)M_7\right)
}{\sqrt{1-r_1^2r_2^2}}
\right)}\,,
\end{equation*}
where $s(\theta)$ and $c(\theta)$ represent ${\rm{sin}}(\theta)$ and ${\rm{cos}}(\theta)$ respectively.

If we define $h$ by $ g=e^{ih} $ then we can write the connection as
\begin {equation*}
A_k=pg(\partial_k h)g^{-1}\,.
\end {equation*}
As, by abuse of notation, we have adopted the same notation for coordinates of $ S^6$ and $ S^3\times S^3$, 
the pullback by the smash product acts trivially so this same expression is also the connection of our Poincar\'e bundle.   
Finally, the curvature of the Poincar\'e bundle is
\begin {equation*}
F_{jk}=pg [\partial_j h,\partial_k h] g^{-1}\,. 
\end {equation*}

\section*{Acknowledgments}

The research of JE is supported by NFSC MianShang grant 11375201.
The research of PB and VM is supported by ARC Discovery Project grant DP150100008. 
V.M. also acknowledges support from the Australian Laureate Fellowship FL170100020.

% ===================================================================

% ===================================================================

\end{document}